\documentclass[10pt]{amsart}

\usepackage{amssymb,amsthm,amsmath,enumerate}
\usepackage[numbers,sort&compress]{natbib}
\usepackage{color}
\usepackage{graphicx}
\usepackage{amssymb,amsthm,amsmath}
\usepackage[numbers,sort&compress]{natbib}
\usepackage{color}
\usepackage{graphicx}

\title[ ]{  Absence of singular  continuous  spectrum for  perturbed  discrete Schr\"odinger operators }

\author{Wencai Liu}
\address[Wencai Liu]{Department of Mathematics, University of California, Irvine, California 92697-3875, USA}\email{liuwencai1226@gmail.com}

\theoremstyle{plain}
\newtheorem{theorem}{Theorem}[section]

\newtheorem{lemma}[theorem]{Lemma}

\newcommand{\R}{\mathbb{R}}
\newcommand{\C}{\mathbb{C}}
\newcommand{\Z}{\mathbb{Z}}

\theoremstyle{definition}

\begin{document}


\begin{abstract}
We show that   the spectral measure of    discrete  Schr\"odinger  operators $
     (Hu)(n)=  u({n+1})+u({n-1})+V(n)u(n)$
does not have singular continuous component   if the potential $V(n)=O(n^{-1})$.

\end{abstract}
\maketitle
\section{Introduction and main results}

We consider the  discrete Schr\"odinger  operator   on $\ell^2(\Z^+)$,
\begin{equation}\label{Gdis}
  (Hu)(n)=  u({n+1})+u({n-1})+V(n)u(n),
\end{equation}
 where $V(n)$ is the potential.

Denote by $H_0$ the free discrete Schr\"odinger  operator on $\ell^2(\Z^+)$.
Without loss of generality, we assume the operator given by \eqref{Gdis} satisfies the Dirichlet  boundary condition at zero.

In this paper, we are interested in the spectral theory of $H_0+V $ with power-decaying potentials:
\begin{equation*}
  |V(n)|\leq\frac{O(1)}{1+n^{\alpha}}
\end{equation*}
for some $\alpha>0$.

We also introduce the continuous Schr\"odinger  operator on $L^2(\R^+)$, namely,
\begin{equation*}
  Hu=-u^{\prime\prime}+Vu
\end{equation*}
with $|V(x)|\leq\frac{O(1)}{1+x^{\alpha}}$.

  Schr\"odinger  operators with power decay potentials have attracted a lot of attentions and achieved a  remarkable progress.
Roughly speaking, $\alpha=\frac{1}{2}$ is the sharp transition for $\sigma_{ac}(H_0+V)=\sigma_{ac}(H_0)$  and $\alpha=1$ is the sharp transition for absence of (singular continuous spectrum) embedded eigenvalues. We refer   readers to  a  survey paper \cite{Kiselevsimon} for the progress in this area.

Let us go back to the discrete model. If $V(n)=\frac{o(1)}{n}$, $\sigma_{pp}(H_0+V)\cap(-2,2)=\emptyset$. Wigner-von Neumann  type functions
    $V(n)=\frac{c}{1+n}\sin(  kn+\phi)$ show that $H_0+V$ may have eigenvalues in $(-2,2)$ if we allow $V(n)=\frac{O(1)}{1+n}$. See \cite{liu2018criteria} for the quantitative results.
    For the singular continuous spectrum,  Remling proved that $\sigma_{sc}(H_0+V)=\emptyset$ if $V(n)=\frac{o(1)}{n}$ \cite{remlingsharp}.  In this paper, we obtain
 \begin{theorem}\label{Thm1:absencesc}
 Suppose the potential $V(n)$ satisfies $\limsup_{n\to \infty}n|V(n)|<\infty$. Then
 the operator $H_0+V$  does not have singular continuous spectrum.
 \end{theorem}
One of our motivations is from the continuous Schr\"odinger  operator.
For the continuous case, Kiselev proved that $\sigma_{sc}(H_0+V)=\emptyset$ if $V(x)=\frac{O(1)}{1+x}$ and for any  given any positive function  $h(x)$ tending to infinity as  $x\to \infty$,
there exist potentials $V (x) $ such that $|V(x)|\leq \frac{h(x)}{1+x}$
 and the singular continuous spectrum of the operator $H_0+V$ is non-empty \cite{Kiselevjams05}.
  It is natural to ask whether such sharp spectral transitions  hold for   discrete cases or not. In this note, we prove that  the absence of the singular continuous spectrum  is still true for   discrete cases. We conjecture here that $|V(n)|= \frac{O(1)}{1+n}$  is the sharp transition for absence for singular continuous spectra.
In the forthcoming paper, the author will study  the same topic of perturbed periodic operators \cite{liusc2}.
  Comparing to continuous cases, the spectral properties of discrete cases strongly depend on the arithmetic properties of the  quasimomentum \cite{liu2018criteria} and
  the Pr\"ufer angle is evolved  in a singular way (there is a $\cot$ function involved ).
Because of those difficulties,  the spectral features of discrete   operators  are usual much more delicate than those of continuous cases.
   For example, the sharp transition for single embedded eigenvalues for the continuous case  was known  forty years ago   \cite{atk}. However, the sharp transition for single
  embedded eigenvalues for the discrete case was partially solved   by the author only a short time ago \cite{liu2018criteria}.
  The construction of potentials with  dense embedded eigenvalues for perturbed periodic operator was known for around 20 years \cite{KRS}. However, similar results for the discrete case
  were only done in  very recent papers \cite{nabokostu18,ld}.
Although the proof of this paper follows  the strategy for the continuous case \cite{Kiselevjams05},   the  extension is not completely straightforward.

 In the following, we always assume that
 \begin{equation}\label{Gbdp}
    |V(n)|\leq \frac{B}{1+n},
\end{equation}
for some $B>0$.
 \section{Preliminaries}\label{SePre}

For $z\in \C\backslash \R$, denote by $\tilde{v}(n,z)$ ($\tilde{u}(n,z)$) the solution of \eqref{Gdis} with boundary condition $\tilde{v}(0,z)=1$ and $\tilde{v}(1,z)=0$ ($\tilde{u}(0,z)=0$ and $\tilde{u}(1,z)=1$).
The Weyl $m$-function $m(z)$ (well defined on $z\in \C\backslash \R$) is given by the unique complex number $m(z)$ so that  $\tilde{v}(n,z)+m(z)\tilde{u}(n,z)\in\ell^2(\Z^+)$.
The spectral measure $\mu$ on $\R$, is given by the follow formula, for $z\in \C\backslash \R$
\begin{equation*}
  m(z)=\int\frac{d\mu(x)}{x-z}.
\end{equation*}
Denote $\mu_{sc}$ by the singular continuous component of $\mu$. It is well known that $\sigma_{\rm sc}(H_0+V)=\emptyset $ if and only if
$\mu_{sc}=0$.

By Weyl law, $\sigma_{\rm ess}(H)=(-2,2)$. In order to prove Theorem \ref{Thm1:absencesc},  it suffices to show
$\mu_{sc}(-2,2)=0$.

 For any $E\in(-2,2)$, let $E=2\cos \pi k$ with $k\in(0,1)$. We mention that $k$ depends on $E$. However, we omit the dependence for simplicity.
By symmetry, we only need to show there is no sc component in $(0,2)$.
Fix any   closed interval  $I$  in $(0,2)$,
define $\tilde{I}=\{k(E): E=2\cos\pi k(E)\in  {I}\}$     so that $\tilde{I}$ is  a closed interval in  $ (0,\frac{1}{2})$.
In the following, we always assume $E\in  {I}$ ($k\in \tilde{I}$).

 Let us introduce the Pr\"ufer transformation first (cf. \cite{remlingsharp,KLS,KRS}).
Suppose $u(n,E)$ (sometimes we also use $u(n,k)$) is a solution of \eqref{Gdis} with $u(0,E)=0$ and $ u(1,E)=1$.

Let
\begin{equation}\label{L2}
    Y(n,k)=\frac{1}{\sin \pi k} \left(
                                  \begin{array}{cc}
                                    \sin \pi k & 0 \\
                                    -\cos \pi k & 1 \\
                                  \end{array}
                                \right)\left(\begin{array}{c}
                                         u(n-1,k) \\
                                         u(n,k)
                                       \end{array}\right).
\end{equation}

Define the Pr\"ufer variables $R(n,k)$ and $\theta(n,k)$ as
\begin{equation}\label{L21}
    Y(n,k)=R(n,k)\left(\begin{array}{c}
                                        \sin(\pi \theta(n,k)-\pi k) \\
                                        \cos(\pi \theta(n,k)-\pi k)
                                       \end{array}\right).
\end{equation}
It is well known that $R$ and $\theta$ obey the equations
\begin{equation}\label{PrufR}
    \frac{R(n+1,k)^2}{R(n,k)^2}=1-\frac{V(n)}{\sin \pi k}\sin 2\pi \theta(n,k)+\frac{V(n)^2}{\sin^2\pi k}\sin^2\pi \theta(n,k)
\end{equation}
and
\begin{equation}\label{PrufT}
    \cot (\pi \theta(n+1,k)-\pi k)=\cot \pi \theta(n,k)-\frac{V(n)}{\sin \pi k}.
\end{equation}
By the Dirichlet boundary condition, we have the initial conditions
\begin{equation*}
    R(0,k)=\frac{1}{\sin \pi k},\theta(0,k)=k.
\end{equation*}
We will give several Lemmas, which will be used in the following sections.
\begin{lemma}\cite[Prop.2.4]{KLS}\label{Leap3}
Suppose $\theta(n,k)$ satisfies \eqref{PrufT} and $|\frac{V(n)}{\sin \pi k}|<\frac{1}{2}$. Then we have
\begin{equation}\label{Gap3}
 | \theta(n+1,k)- k-\theta(n,k)|\leq \left|\frac{V(n)}{\sin \pi k}\right|.
\end{equation}
\end{lemma}

\begin{lemma}\cite[Lemma 4.4]{KLS}\label{Lee}
  Let $\{e_i\}_{i=1}^N$ be a set of unit vector in a Hilbert space $\mathcal{H}$  so that
  \begin{equation*}
    \alpha=N\sup_{i\neq j}| \langle e_i,e_j\rangle|<1.
  \end{equation*}
  Then
  \begin{equation}\label{Gapr71}
    \sum_{i=1}^N|\langle g,e_i\rangle|^2\leq (1+\alpha)||g||^2.
  \end{equation}
  \end{lemma}
  For  $L\in\Z^+$, let $ V_L$ be the cut off $V$ up to $L$. Namely, $V_L(n)=V(n)$ for $0\leq n \leq L$ and $V_L(n)=0$ for $n>L$.
Let $ \mu_L$ be the spectral measure corresponding to the operator with potential $V_L$.
\begin{lemma}\label{Letwomu}
\cite{GKT}
Fix some compact interval $I\in (-2,2)$ and arbitrary $M,\sigma>0$. Then for any $\epsilon>L^{-\frac{1}{1+\sigma}}$, we have
\begin{equation}\label{Gtwomu}
   \mu(E-\epsilon,E+\epsilon)\geq \mu_L(E-\frac{\epsilon}{2},E+\frac{\epsilon}{2})-C(I,\sigma,B,M)\epsilon^M
\end{equation}
for any $(E-\epsilon,E+\epsilon)\subset (-2,2)$.
\end{lemma}

\begin{lemma}\cite{remlingsharp}\label{Lezero}
Under the assumption of \eqref{Gbdp}, the spectral measure $\mu$ of $H=H_0+V$ is zero dimensional.
\end{lemma}
\section{ Technical Lemmas}

  \begin{lemma} \label{Lealmost}
  For any $k\in \tilde{I}$, we have
\begin{equation}\label{Gcons4}
  \left|  \sum _{n=1}^{L} \frac{\cos 4 \theta(n,k)}{n}\right|\leq C(I,B).
\end{equation}
For any $k_1,k_2\in \tilde{I}$ and $k_1\neq k_2$, we have
\begin{equation}\label{Gcons5}
    \left|\sum_{n=1}^{L}\frac{\sin 2 \theta(n,k_1) \sin 2 \theta(n,k_2)}{n}\right|\leq  C(I,B)\log(|k_1-k_2|^{-1})+C(I,B).
\end{equation}

\end{lemma}
\begin{proof}
We start with the proof of \eqref{Gcons4}.  It suffices to show
\begin{equation*}
\left|  \sum _{n=1}^L \frac{e^{ 4 i\theta(n,k)}}{n}\right|\leq C(I,B).
\end{equation*}
Straightforwardly,
\begin{eqnarray}
\nonumber
   \left|(e^{4\pi i k}-1)\sum_{n=1}^L \frac{e^{4 i\theta(n,k)}}{n}\right|
   &=& \left|\sum_{n=1}^L \frac{e^{4 i(\theta(n,k)+k)}}{n}-\sum_{n=1}^L \frac{e^{4 i\theta(n,k)}}{n}\right| \\\nonumber
   &=&  \left|\sum_{n=1}^L \frac{e^{4 i\theta(n+1,k)}}{n}-\sum_{n=1}^L \frac{e^{4 i\theta(n,k)}}{n}+\sum_{n=1}^L \frac{e^{4 i(\theta(n,k)+k)}}{n}-\sum_{n=1}^L \frac{e^{4 i\theta(n+1,k)}}{n}\right|\\\nonumber
    &\leq&  \left|\sum_{n=1}^L \frac{e^{4 i\theta(n+1,k)}}{n}-\sum_{n=1}^L \frac{e^{4 i\theta(n,k)}}{n}\right|+\left|\sum_{n=1}^L \frac{e^{4 i(\theta(n,k)+k)}}{n}-\sum_{n=1}^L \frac{e^{4 i\theta(n+,k)}}{n}\right|\\
     &\leq& 2+ \left|\sum_{n=1}^{L-1} (\frac{1}{n}-\frac{1}{n+1})e^{4 i\theta(n+1,k)}\right|+\left|\sum_{n=1}^L \frac{e^{4 i(\theta(n,k)+k)}}{n}-\sum_{n=1}^L \frac{e^{4 i\theta(n+1,k)}}{n}\right|.\label{Gequ1}
\end{eqnarray}
By \eqref{Gap3},  \eqref{Gequ1} and $|e^{4\pi i k}-1|=2|\sin2\pi k|$, we have
\begin{equation*}
    \left|\sum_{n=1}^L \frac{e^{4 i\theta(n,k)}}{n}\right| \leq C(I,B).
\end{equation*}
Now we are in the position to prove \eqref{Gcons5}.
Trigonometric identity implies
\begin{equation}\label{Gequ2}
   2 \sin 2 \theta(n,k_1) \sin 2 \theta(n,k_2)=\cos 2( \theta(n,k_1)- \theta(n,k_2))-\cos 2( \theta(n,k_1)+ \theta(n,k_2)).
\end{equation}
By the same proof of \eqref{Gcons4}, one has
\begin{equation*}
    \left|\sum_{n=1}^L \frac{\cos 2( \theta(n,k_1)+ \theta(n,k_2))}{n}\right| \leq C(I,B).
\end{equation*}
It suffices to show
\begin{equation*}
    \left|\sum_{n=1}^L \frac{e^{ 2i \theta(n,k_1)-2i\theta(n,k_2)}}{n}\right| \leq C(I,B)+C(I,B)\log (|k_1-k_2|^{-1}).
\end{equation*}
Since  $$\left|\sum_{n=1}^{|k_1-k_2|^{-1}} \frac{e^{ 2i \theta(n,k_1)-2i\theta(n,k_2)}}{n}\right| \leq C(I,B)+C(I,B)\log (|k_1-k_2|^{-1}),$$
we only need to prove
\begin{equation*}
    \left|\sum_{n=|k_1-k_2|^{-1}}^L \frac{e^{ 2i \theta(n,k_1)-2i\theta(n,k_2)}}{n}\right| \leq C(I,B).
\end{equation*}
By the proof of \eqref{Gequ1}, we have
\begin{eqnarray*}
   \left|(e^{2\pi i(k_1-k_2)}-1)\sum_{n=|k_1-k_2|^{-1}}^L \frac{e^{ 2i \theta(n,k_1)-2i\theta(n,k_2)}}{n}\right| &\leq& C(I,B)\sum_{n=|k_1-k_2|^{-1}}^L \frac{1}{n^2} \\
   &\leq& C(I,B) |k_1-k_2|
\end{eqnarray*}
It leads to
\begin{equation*}
 \left| \sum_{n=|k_1-k_2|^{-1}}^L \frac{e^{ 2i \theta(n,k_1)-2i\theta(n,k_2)}}{n}\right| \leq C(I,B).
\end{equation*}
We finish the proof.
\end{proof}

\begin{lemma}\label{Lemu}
The following formula hold,
\begin{equation}\label{Gmul}
   \frac{d\mu_L(E)}{d E}=\frac{ 1}{\pi\sin \pi k}\frac{1}{R^2(L+1,E)}
\end{equation}
for $E\in(-2,2)$.

\end{lemma}
\begin{proof}

Let $z=E+i\varepsilon$ for $E\in(-2,2)$ and $\varepsilon>0$. Let $ k(z)+i \gamma(z)$ be such that $2\cos \pi (k(z)+i\gamma(z))=z$ with $ k(z)\in \R$ and $\gamma(z)\in \R $.  Thus
\begin{equation*}
  (  e^{-\pi \gamma}+e^{\pi \gamma})\cos\pi k=E; (  e^{-\pi \gamma}-e^{\pi \gamma})\sin\pi k=\varepsilon.
\end{equation*}
Let us choose the branch  so that $k(z)\in(0,1)$ and $ \gamma(z)<0$.
It is easy to see
\begin{equation*}
    \lim_{\varepsilon\to 0+}k(E+i\varepsilon)=k(E),    \lim_{\varepsilon\to 0+}\gamma(E+i\varepsilon)=0
\end{equation*}
where $2\cos\pi k(E)=E$ with  $k(E)\in(0,1)$.

Define $ \tilde{u}(n,z)=e^{-i \pi (k+i\gamma)n}$ for $n\geq L$ and extend $\tilde{u}(n,z)$ to $0\leq n\leq L$  by solving equation
\begin{equation*}
  \tilde{u}(n+1,z)+\tilde{u}(n-1,z)+(V_L(n)-z)\tilde{u}(n,z)=0
\end{equation*}
for $0\leq n\leq L-1$.  Since $ \gamma(z)<0$, one has  $\tilde{u}(n,z)\in \ell^2(\Z^+)$.
By  spectral theory (we refer the readers to \cite{Simon} and references therein for details), we have
\begin{equation*}
    m(z)=-\frac{\tilde{u}(1,z)}{\tilde{u}(0,z)},
\end{equation*}
and
\begin{equation}\label{Gequ13}
    \frac{d\mu_L}{dE}=\frac{1}{\pi}\lim_{\varepsilon\to 0+} \Im m(E+i\varepsilon).
\end{equation}

Let $T(z)$ be the transfer matrix of $H_0+V_L$ from $0$ to $L$, that is
\begin{equation*}
  T (z)\left(\begin{array}{cc}
               \phi(0) \\ \phi (1)
              \end{array}
  \right)=\left(\begin{array}{cc}
                \phi(L) \\ \phi(L+1)
              \end{array}
  \right)
\end{equation*}
for any solution $\phi$ of $(H_0+V_L)\phi=z\phi$.

Let
\begin{equation*}
  T(z)=\left(
           \begin{array}{cc}
             a(z) & b(z) \\
             c(z) & d(z) \\
           \end{array}
         \right).
\end{equation*}
Thus
\begin{eqnarray*}
  \left(\begin{array}{cc}
                \tilde{u}(0,z) \\ \tilde{u} (1,z)
              \end{array}
  \right) &=&\left(
           \begin{array}{cc}
             a(z) & b(z) \\
             c(z) & d(z) \\
           \end{array}
         \right)^{-1}\left(\begin{array}{cc}
                \tilde{u}(L,z) \\ \tilde{u} (L+1,z)
              \end{array}
  \right) \\
   &=& \left(
           \begin{array}{cc}
             d(z) & -b(z) \\
             -c(z) & a(z) \\
           \end{array}
         \right) \left(\begin{array}{cc}
                \tilde{u}(L,z) \\ \tilde{u} (L+1,z)
              \end{array}
  \right).
\end{eqnarray*}
Direct computation implies that
\begin{eqnarray}
 \nonumber   \lim_{\varepsilon\to 0+} \Im m(E+i\varepsilon) &=&- \Im \frac{ ae^{-i\pi k}-c}{d-be^{-i\pi k}} \\
   &=& \frac{\sin \pi k}{(d-b\cos\pi k)^2+b^2\sin^2\pi k}.\label{Gequ12}
\end{eqnarray}
It is easy to see that
\begin{eqnarray*}
  \left(\begin{array}{cc}
                u(L) \\ u (L+1)
              \end{array}
  \right) &=&T(E)\left(\begin{array}{cc}
                u(0) \\ u(1)
              \end{array}
  \right) \\
   &=& T(E)\left(\begin{array}{cc}
               0\\ 1
              \end{array}
  \right)= \left(\begin{array}{cc}
                b \\ d
              \end{array}
  \right).
\end{eqnarray*}
By \eqref{L2} and \eqref{L21}, one has
\begin{equation}\label{Gequ14}
   \frac{1}{R^2(L+1,E)}= \frac{\sin^2 \pi k}{(d-b\cos\pi k)^2+b^2\sin^2\pi k}.
\end{equation}
Now the Lemma follows from \eqref{Gequ13}, \eqref{Gequ12} and \eqref{Gequ14}.
\end{proof}
\section{Proof of Theorem \ref{Thm1:absencesc}}
Once we have Lemmas \ref{Lealmost} and  \ref{Lemu} at hand, Theorem  \ref{Thm1:absencesc} can be proved in a similar way as the argument in \cite{Kiselevjams05}.
For convience, we give all the details here.

Fix $0<\beta<1$, $M=1+\beta$ and $\sigma>0$. We will choose small enough $\epsilon>0$ (depends on $B$, $\beta$, $M>1$ and $\sigma>0$). Let $L=\lfloor\epsilon^{-1-\sigma}\rfloor$, where $ \lfloor x\rfloor$ is the integer part of $x$.
Let $C_1=C_1(B, I)$, which will be determined later.

We say a subset $ S\subset I$ is $\epsilon-N$ separate, if the following two conditions hold:

For any $k\in S$,
\begin{equation}\label{Gassum1}
    |\sum_{n=1}^LV(n)\sin2\theta(n,k)|\geq (1-\beta)C_1(B,I)\log \epsilon^{-1}.
\end{equation}
For any $k_1,k_2\in S$ and $k_1\neq k_2$,
\begin{equation}\label{Gassum2}
   |k_1-k_2|\geq \epsilon^{1/N^2}.
\end{equation}
\begin{theorem}\label{Thmbound}
There exists $\epsilon_1(B,I,\sigma,\beta)>0$ and $C(B,I,\sigma,\beta)$ such that  for any $\epsilon<\epsilon_1$ and $N\geq C(B,I,\sigma,\beta)$,  the
$\epsilon-N$ separate  set $S$ satisfies $\# S\leq N$.
\end{theorem}
\begin{proof}
We consider the Hilbert space
\begin{equation*}
  \mathcal{H}=\{u\in \R^{L}:\sum_{n=1}^{L}n|u(n)|^2 <\infty\}
\end{equation*}
with the inner product
\begin{equation*}
    \langle u,v \rangle=\sum_{n=1}^{L} u(n)v(n)n.
\end{equation*}
In  $\mathcal{H}$, by \eqref{Gbdp} we have
\begin{equation}\label{Gapr77}
  ||V||_{ \mathcal{H}}^2\leq B^2\log L.
\end{equation}
Let
\begin{equation*}
  e_{i}(n)=\frac{1}{\sqrt{A_i}}\frac{\sin 2\theta(n,k_i)}{n}\chi_{[1,L]}(n),
\end{equation*}
where $A_i$ is chosen so that $e_i$ is a unit vector in $\mathcal{H}$.
We have the following estimate,
\begin{eqnarray}
  A_i  &=& \sum_{n=1}^{L}\frac{\sin^2 2\theta(n,k_i)}{n} \nonumber\\
   &=&\sum_{n=1}^{L}\frac{1}{2n}- \sum_{n=1}^{L}\frac{\cos  4\theta(n,k_i)}{2n}\nonumber.
\end{eqnarray}
By \eqref{Gcons4}, one has
\begin{equation}\label{Gequ4}
    | A_i- \frac{1}{2}\log L|\leq C(I,B)
\end{equation}
By \eqref{Gcons5} and \eqref{Gequ4}, we have
\begin{equation}\label{Gapr78}
 | \langle e_i ,e_{ j}  \rangle\leq \frac{C(B,I)}{1+\sigma}N^{-2} +\frac{C(B,I)}{(1+\sigma)\log \epsilon^{-1}}.
\end{equation}
The first condition  \eqref{Gassum1} implies
\begin{equation}\label{Gequ5}
    |\langle V,e_i \rangle|^2\geq \frac{(1-\beta)^2C_1^2}{1+\sigma}\log \epsilon^{-1}.
\end{equation}
By \eqref{Gapr71} and \eqref{Gapr78}, one has
\begin{equation}\label{Gapr711}
\sum_{i=1}^N  |\langle V,e_i\rangle_{\mathcal{H}}|^2\leq \left(1+\frac{C(B,I)}{1+\sigma}N^{-1} +\frac{NC(B,I)}{(1+\sigma)\log \epsilon^{-1}}\right)||V||_{\mathcal{H}}.
\end{equation}
By \eqref{Gapr77}, \eqref{Gequ5}  and \eqref{Gapr711}, we have
\begin{equation*}
    N\left(  \frac{C_1^2(1-\beta)^2}{1+\sigma}\log \epsilon^{-1}\right)\leq  \left(1+\frac{C(B,I)}{1+\sigma}N^{-1} +\frac{NC(B,I)}{(1+\sigma)\log \epsilon^{-1}}\right) B^2(1+\sigma)\log \epsilon^{-1}.
\end{equation*}
This implies the Lemma.
\end{proof}
Assume that the singular continuous spectrum is not empty.
As the analysis in  the beginning of  \S \ref{SePre}, there exists $\delta>0$ such that  $\mu_{sc}(I)=\delta$. Fix a small number  $\epsilon$ and a large number  $N$ such that Theorem \ref{Thmbound} holds.
By making $\epsilon $ smaller and the continuity of $\mu_{sc}$, we assume
 $\mu_{sc}(J) <\frac{1}{32}\delta N^{-3}$ for
any interval $J\subset I$ such that $|J|\leq \epsilon^{N^{-2}}$.

Let $m\in \Z^+$.
 We say that an interval $J\subset I$ belongs
to the scale $m$ if $|J| \leq \epsilon_m\equiv\epsilon^m$.
 We call an interval $J$ of scale $m$ singular if $\mu_{sc}(J)\geq \epsilon_m^{\beta}$.
We call two intervals of the scale $m$ separated if the distance between
their centers exceeds $2\epsilon_m^{N^{-2}}$.
 \begin{lemma}\label{Leinse}
 There can be no more than $N$ separated singular intervals at each
scale.
 \end{lemma}
 \begin{proof}
 Assume that $J^m_l$, $l = 1, \cdots , N$ are separated singular intervals of scale $m$.
    Let $L_m=\rfloor\epsilon_m^{-1-\sigma}\rfloor$ and denote  by $\mu_m$  the spectral measure corresponding to the
potential being cut off at $L_m$. Denote by $2J^m_l$
the interval with the same
center as $J^m_l$ but twice its size. Then by \eqref{Gmul} we have
\begin{equation}\label{Gequ7}
     \mu_m(2J^m_l)\geq \mu(J^m_l)-C(B,I,\sigma,\beta)\epsilon_m^M\geq \frac{1}{2}\epsilon_m^{\beta},
\end{equation}
provided $\epsilon$ is small enough.
  Combining \eqref{Gmul} with \eqref{Gequ7}, we see that there exist
$k^m_l\in 2J^m_l$ such that
\begin{equation}\label{PrufRau1}
    R^{2}(L_m,k_j^m)\leq C(I)\epsilon_m^{1-\beta}.
\end{equation}
We will show that $k_j^m$, $j=1,2,\cdots,N$ is $\epsilon_m-N$ separate.
By \eqref{PrufR}, one has
\begin{equation}\label{PrufRau}
   \ln R(L,k)^2- \ln R(1,k)^2=-\sum_{n=1}^L\frac{V(n)}{\sin \pi k}\sin 2\pi \theta(n,k) +O(1).
\end{equation}
By \eqref{PrufRau1} and  \eqref{PrufRau}, we have that
the assumption \eqref{Gassum1} holds for suitable $C_1$.
 Moreover, by the separation
assumption of scale $m$, $|k_i^m-k_j^m|>\epsilon_m^{N^{-2}}$, which implies \eqref{Gassum2}.
 Now the Lemma follows from Theorem \ref{Thmbound}.
 \end{proof}
\begin{proof}[\bf Proof of Theorem \ref{Thm1:absencesc}]
Define the set $S_m$ as a union of
all singular intervals $J$ at scale $m$. By Lemma \ref{Leinse}, it is easy to see that the set $S_m$ can be covered
by at most $8N$ intervals of size  $\epsilon_m^{N^{-2}}$. We denote them by $\tilde{J}_l^m$.  By the smallness choice of $\epsilon$, we have for any $m\in \Z^+$,
\begin{equation*}
    \mu_{sc}(S_m)\leq 8 N\frac{1}{32}\delta N^{-3}=\frac{1}{4}N^{-2}\delta.
\end{equation*}
It  yields that
\begin{equation}\label{Gequ9}
   \sum_{m=1}^{N^2} \mu_{sc}(S_m)\leq \frac{1}{4} \delta.
\end{equation}
 Denote by $\tilde{m}=\lfloor mN^{-2}\rfloor$.
Then any interval $\tilde{J}^m_l$ satisfying
 $\mu_{sc}(\tilde{J}_l^m)\geq\epsilon_{\tilde{m}}^{\beta}$ already belongs to $S_{\tilde{m}}$ since $|\tilde{J}_l^m|=\epsilon_m^{N^{-2}}\leq \epsilon_{\tilde{m}}$. Therefore, for any $m\geq N^2$, we have
 \begin{equation}\label{Gequ10}
    \mu_{sc}(S_m\backslash \bigcup_{l<m}S_l)\leq 8N\epsilon_{\tilde{m}}^{\beta}.
 \end{equation}
 By \eqref{Gequ10} and the fact that  each $\tilde{m}$ has at most $N^2$ corresponding $m$, we have
 \begin{equation}\label{Gequ11}
   \sum _{m=N^2}^{\infty}\mu_{sc}(S_m\backslash \bigcup_{l<m}S_l)\leq \sum_{\tilde{m}=1}^{\infty}8N^3\epsilon_{\tilde{m}}^{\beta}\leq 16N^3\epsilon^{\beta}.
 \end{equation}
 By \eqref{Gequ9} and \eqref{Gequ11}, we finally obtain
 \begin{eqnarray}
    \mu_{sc}(\cup_m S_m)
    &\leq& \frac{\delta}{4} +16N^3\epsilon^{\beta}\leq \frac{\delta}{2},\label{Gequ8}
 \end{eqnarray}
if  $\epsilon$ is small enough  ($\epsilon^{\beta}\leq \frac{\delta}{64N^3}$).
On the other hand, by Lemma \ref{Lezero}, the spectral measure can
only be zero-dimensional.
Thus,  $\mu_{sc}$ is supported on a set $S$ such that for any $E \in S$ and any $\alpha>0$ (see \cite[Corollary 2.2 ]{RJLS96} for example),
\begin{equation*}
    D^{\alpha}\mu_{sc}(E)=\limsup_{\varepsilon\to 0}\frac{\mu_{sc}(E-\varepsilon,E+\varepsilon)}{2^{\alpha}\varepsilon^{\alpha}}=
    \infty.
\end{equation*}
In particular, $S\subset \cup S_m$.
It implies
\begin{equation*}
    \delta=\mu_{sc}(S)=\mu_{sc}(\cup S_m)\leq\frac{1}{2}\delta.
\end{equation*}
This is impossible.

\end{proof}
 \section*{Acknowledgments}
    W.L. was supported by   NSF DMS-1700314. This research was also supported by  NSF DMS-1401204.


\footnotesize
 \begin{thebibliography}{10}

\bibitem{atk}
F.~V. Atkinson and W.~N. Everitt.
\newblock Bounds for the point spectrum for a {S}turm-{L}iouville equation.
\newblock {\em Proc. Roy. Soc. Edinburgh Sect. A}, 80(1-2):57--66, 1978.

\bibitem{RJLS96}
R.~del Rio, S.~Jitomirskaya, Y.~Last, and B.~Simon.
\newblock Operators with singular continuous spectrum. {IV}. {H}ausdorff
  dimensions, rank one perturbations, and localization.
\newblock {\em J. Anal. Math.}, 69:153--200, 1996.

\bibitem{Kiselevsimon}
S.~A. Denisov and A.~Kiselev.
\newblock Spectral properties of {S}chr\"odinger operators with decaying
  potentials.
\newblock In {\em Spectral theory and mathematical physics: a {F}estschrift in
  honor of {B}arry {S}imon's 60th birthday}, volume~76 of {\em Proc. Sympos.
  Pure Math.}, pages 565--589. Amer. Math. Soc., Providence, RI, 2007.

\bibitem{GKT}
F.~Germinet, A.~Kiselev, and S.~Tcheremchantsev.
\newblock Transfer matrices and transport for {S}chr\"odinger operators.
\newblock {\em Ann. Inst. Fourier (Grenoble)}, 54(3):787--830, 2004.

\bibitem{nabokostu18}
E.~Judge, S.~Naboko, and I.~Wood.
\newblock Spectral results for perturbed periodic {J}acobi matrices using the
  discrete {L}evinson technique.
\newblock {\em Studia Math.}, 242(2):179--215, 2018.

\bibitem{Kiselevjams05}
A.~Kiselev.
\newblock Imbedded singular continuous spectrum for {S}chr\"odinger operators.
\newblock {\em J. Amer. Math. Soc.}, 18(3):571--603, 2005.

\bibitem{KLS}
A.~Kiselev, Y.~Last, and B.~Simon.
\newblock Modified {P}r\"ufer and {EFGP} transforms and the spectral analysis
  of one-dimensional {S}chr\"odinger operators.
\newblock {\em Comm. Math. Phys.}, 194(1):1--45, 1998.

\bibitem{KRS}
A.~Kiselev, C.~Remling, and B.~Simon.
\newblock Effective perturbation methods for one-dimensional {S}chr\"odinger
  operators.
\newblock {\em J. Differential Equations}, 151(2):290--312, 1999.

\bibitem{liusc2}
W.~Liu.
\newblock {WKB} and absence of singular continuous spectrum for perturbed
  periodic {S}chr\"odinger operators.
\newblock {\em Preprint}.

\bibitem{liu2018criteria}
W.~Liu.
\newblock Criteria for embedded eigenvalues for discrete {S}chr\"odinger
  operators.
\newblock {\em arXiv preprint arXiv:1805.02817}, 2018.

\bibitem{ld}
W.~Liu and D.~C. Ong.
\newblock Sharp spectral transition for eigenvalues embedded into the spectral
  bands of perturbed periodic operators.
\newblock {\em arXiv preprint arXiv:1805.01569}, 2018.

\bibitem{remlingsharp}
C.~Remling.
\newblock The absolutely continuous spectrum of one-dimensional {S}chr\"odinger
  operators with decaying potentials.
\newblock {\em Comm. Math. Phys.}, 193(1):151--170, 1998.

\bibitem{Simon}
B.~Simon.
\newblock Analogs of the {$m$}-function in the theory of orthogonal polynomials
  on the unit circle.
\newblock {\em J. Comput. Appl. Math.}, 171(1-2):411--424, 2004.

\end{thebibliography}
\end{document}